\documentclass[ejs,authoryear, 12pt]{imsart}
\usepackage[a4paper,left=28mm,right=25mm,top=30mm,bottom=30mm]{geometry}
\usepackage[small]{caption}

\RequirePackage[colorlinks,citecolor=blue,urlcolor=blue]{hyperref}
\RequirePackage{hypernat}
\usepackage{amsbsy}
\usepackage{graphicx}
\RequirePackage{amsthm,amsmath,amssymb}
\usepackage{latexsym, bm}
\usepackage{here}
\usepackage{array}

\newcommand{\inv}{\mathrm{inv}}

\newcommand{\txt}{\textstyle}

\newcommand{\T}{^{\mathrm{T}}}

\newtheorem{prop}{Proposition}

\newtheorem{defn}{Definition}

\newtheorem{lem}{Lemma}

\newcommand{\ci}{\mbox{\protect{ $ \perp \hspace{-2.3ex}
\perp$ }}}

\newcommand{\dep}{\pitchfork}  

\newcommand{\n}[0]{\hspace*{.35em}}

\newcommand{\nn}[0]{\hspace*{.7em}}

\newcommand{\fourl}[0]{\hspace*{1.4em}}

\newcommand{\fla}{\mbox{$\hspace{.05em} \prec
\!\!\!\!\!\frac{\nn \nn}{\nn}$}}

\newcommand{\fra}{\mbox{$\hspace{.05em} \frac{\nn
\nn}{\nn
}\!\!\!\!\! \succ \! \hspace{.25ex}$}}

\renewcommand{\baselinestretch}{1.2}

\pubyear{2013}
\volume{0}
\issue{0}
\firstpage{1}
\lastpage{}

\begin{document}

\begin{frontmatter}

\title{Star graphs induce tetrad correlations:\\  for Gaussian as well as for binary variables}
\runtitle{Star graphs induce tetrad correlations}
\thankstext{T1}{The work (also on ArXiv 1307.5396) reported in this paper  was undertaken during the first author's tenure of a Senior Visiting Scientist Award by the International Agency of Research on Cancer, Lyon.}

\begin{aug}
\author{\fnms{Nanny} \snm{Wermuth}\ead[label=e1]{wermuth@chalmers.se}},  \and  \author{\fnms{Giovanni M.} \snm{Marchetti}\ead[label=e2]{giovanni.marchetti@disia.unifi.it}}
\address{Mathematical Statistics, Chalmers University of Technology, Gothenburg, Sweden;
Medical Psychology and Medical Sociology, Gutenberg-University, Mainz, Germany\\
Dipartimento di Statistica, Informatica, Applicazioni,  ``G. Parenti'', 
University of Florence, Italy\\
\printead{e1}, \printead{e2}}\n\\

  \runauthor{N. Wermuth and G.M. Marchetti}

\end{aug}

\begin{abstract}
Tetrad correlations were obtained historically for Gaussian distributions when tasks  are designed to measure an ability or attitude so that a single unobserved  variable may generate the observed, linearly increasing dependences among the tasks.  We connect such generating processes to a particular type of  directed  graph, the star graph,  and to the notion of traceable regressions. Tetrad correlation conditions  for the existence of a single latent variable are derived.  These are needed for  positive dependences not only in joint Gaussian but also in  joint binary distributions. Three applications with binary items are  given. 
\end{abstract}


\tableofcontents
\end{frontmatter}
\headsep=15pt
\renewcommand{\baselinestretch}{1.2}

\section{Introduction}
Since the seminal work by Bartlett (1935) and  Birch (1963), viewed in  a larger  perspective by Cox (1972), correlation coefficients were barely used as measures of dependence for categorical data.
Instead, functions of the odds-ratio emerged as the relevant  parameters in log-linear probability models for joint distributions and in logistic models, that is  for regressions with a binary response. Compared with other possible measures of dependence, the outstanding advantage of  functions of the odds-ratio is their variation independence 
of the marginal counts: odds-ratios are unchanged under different types of sampling schemes that result by fixing either  the total number, $n$, of observed  individuals or the counts at the levels of one of the variables; see Edwards (1963).

As one consequence of the importance of odds-ratios for discrete random variables, it is no longer widely  known that  Pearson's simple, observed  correlation coefficient, $r_{12}$ say,  coincides in $2 \times 2$ contingency tables with the so-called Phi-coefficient, so that $\sqrt{n}r_{12}$ is, asymptotically and under independence of the two binary variables, a realization of a standard Gaussian distribution. As we  shall see, some  properties of correlation coefficients 
for binary variables, make them important  for  data generating processes  that incorporate  many conditional independences.

In particular, we look here at   {\bf directed star graphs} such as the one  shown in Figure \ref{fig:1}(a). Such graphs have one inner node, $L$, from which  $Q$ arrows start and point  to the uncoupled, outer nodes
$1, \dots, Q$.  To simplify notation, the inner node $L$ denotes also the corresponding random  variable and  both, the node and the variable,  are called {\bf root}. The random variables $X_i$ corresponding to the outer nodes, also called the {\bf leaves} of the graph,  are identified just by their index $i$ taken from $\{1, \dots, Q\}$. 
\begin{figure}[H]
\centering
\includegraphics[scale=.44]{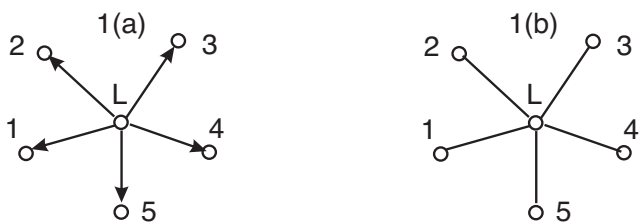}
\caption{A directed star graph (a) and the Markov equivalent undirected star graph (b).}  \label{fig:1}
\end{figure}

The {\bf independence structure} of the directed star graph is  mutual independence of the leaves
given the root. In the condensed node notation, this is written as \begin{equation}
(1 \ci 2\ci \cdots \ci Q) | L. \label{cind}
\end{equation}
In general, the types of variables can be of any kind. Densities, $f_N$, that are said to be generated over the directed star graph with node set $N=\{1,\dots, Q,L\}$,  may also be of any  kind,  provided  that they have the above  independence structure.

Each {\bf generated density} is defined by $Q$ conditional densities, $f_{i|L}$, that are called  regressions, and a marginal density, $f_L$,  of the root. In the condensed node notation, a  joint density with the independence structure  of a directed star graph, factorizes as
\begin{equation}
f_N = f_{1|L} \cdots f_{Q|L}f_L \,.\label {factden}
\end{equation}

Directed star graphs belong to the  class  of  {\bf regression graphs} and to their subclass of directed acyclic graphs.  Distributions generated over regression graphs have been named and studied as  {\bf sequences of regressions} and one  knows   when two regression graphs capture the same independence structure, that is when  they are {\bf Markov equivalent}, even though defined differently; see Theorem 1 in Wermuth and Sadeghi (2012).
In particular, each directed star graph is  Markov equivalent to an undirected star graph with the same node and edge set,
 such as  in Figure  \ref{fig:1}(b), since it does not contain  a {\bf collision {\sf V}}: $\circ \fra \circ \fla \circ$, that is two uncoupled arrows meeting head-on.

{\bf Sequences of regressions are traceable} if  one can use the graph alone to trace pathways of dependence, that is to decide when a non-vanishing  dependence is induced for an uncoupled node pair. For this,  each edge present in the  regression graph  corresponds to a non-vanishing, conditional dependence that is considered to be  strong in a given substantive context.   In evolutionary biology,  the required changes in dependences  that are strong enough  to lead to mutations,  were   for instance called `drastic' by Neyman (1971).
In general, special properties of the generated distributions are required in addition; see Wermuth (2012). 

In the last century,  regression  graphs  were not yet defined and properties of  traceable regression unknown. Then,  distributions generated over directed star graphs, in which the root $L$ corresponds to  a variable that is never directly observed, have been studied separately 
 under different distributional assumptions and with changing main objectives. Many of them were named {\bf item response models}. In these contexts, the  unobserved root $L$ is also  called {\bf latent or hidden}, the {\bf items} are the observed variables.

For instance, Spearman (1904) suggested  to measure 
general intelligence with $Q$ similar quantitative tasks, a method  now called  {\bf factor analysis with a single latent variable}. For confirmatory factor analyses, typically,  a Gaussian  distribution is assumed, after the observed variables are standardized to have  mean zero and unit variance. Heywood (1931) and Anderson and Rubin (1956)
derived necessary and sufficient conditions for the existence of one or more latent variables,
without the constraint of  non-vanishing,   positive dependences of each item on the latent variables. The latter  assumption arises however naturally when the items are to measure a specific latent
ability or attitude or when they are the symptoms of a given disease.

For instance, when psychologists try to measure for  children in a given age range,  what is called the  working memory capacity,  then each item is the successful repetition -- in reverse order -- of  a sequence of numbers.  Typical tasks of the same difficulty are the sequences  (3, 5,  7) and (2, 4, 6). The item difficulty increases, for instance, by presenting more numbers or
numbers of two digits. For children with a higher capacity of the working memory, one expects more successes for tasks of the same difficulty as well as  for tasks of increasing difficulty. 

When instead,  the  leaves and the root of  the star graph are both categorical,  the resulting model is {\bf a latent class model}, as   proposed by
Lazarsfeld (1950), again with an extensive follow-up literature.
An important warning was given  by
Holland and Rosenbaum (1986): such a model can never be falsified when the number of
levels of the latent variable is large compared to the number of cells in the observed contingency table. Expressed in other words, merely requiring conditional independence of categorical items given the latent variable  imposes then  no constraints on the observed distributions. 

A general, testable constraint suggested by Holland and Rosenbaum (1986), that is now widely adopted  in nonparametric item response theory, is to have a monotonically  non-decrea\-sing association of each item on $L$; see van der Ark (2012) or Mair and Hatzinger (2007). However, the underlying notion of `conditional association', that had been proposed  in probability theory,   includes conditionally independent variables as being  `conditionally associated'; see Esary, Proschan and Walkup (1967). By contrast, when one uses
traceable regressions, each edge present in
a  graph  excludes explicitly a corresponding  conditional independence statement.

In this paper, we  study similarities of  joint Gaussian and of binary  distributions generated over  directed star
graphs  where  all dependences of leaves, $1, \dots, Q$, on the root, $ L$,  are positive and the root is unobserved but does not coincide with any  leaf or with any combination of the leaves.  In Section 2,
we summarize  results for star graphs and for their traceable regressions. In Section 3, we describe  joint Gaussian distributions, so generated, and  in Section 4, we  study  joint binary distributions, especially correlation constraints on the  distribution of the leaves. Section 5,  gives applications to binary distributions, Section 6 a general discussion. The Appendix  contains a technical proof.

\section{Marginalizing over the root in star graphs}

A regression graph  is said to be {\bf edge-minimal} when each of its edges, that is present in the graph, indicates a non-vanishing dependence. 
\begin{defn} {\bf Traceable regressions} are generated over an edge-minmal, directed star graph, if (a)
 the density factorizes as in equation  \eqref{factden} and   (b) is {\bf dependence-inducing} by
marginalizing over the root $L$, that is yields  for each pair of leaves  $i,j $ a bivariate dependence, denoted by $i\dep j$.   \end{defn}

  Thus, for traceable regressions with exclusively strong dependences, $i\dep L$, in the generating star graph, each $ij$-edge in  the induced   {\bf complete covariance graph} of the leaves, drawn  as   in Figure~\ref{fig:2}(a), indicates a non-vanishing dependence, $i \dep j$. 
Probability distributions that do not induce such a dependence violate {\bf singleton-transitivity}; see Table 1 in Wermuth (2012)
for an example of such a family of distributions.
 
  More consequences can be deduced, using strong dependences $i\dep L$ and the Markov equivalence of the directed to the undirected star graph with the same node and edge set, as in Figures~\ref{fig:1}(a) and \ref{fig:1}(b). For  a traceable distribution with an undirected star graph, such as in Figure~\ref{fig:1}(b),  each edge in the induced  {\bf complete concentration graph}, obtained by marginalising over the root $L$ and  drawn  as in Figure~\ref{fig:2}(b), indicates  a non-vanishing 
  conditional dependence of each pair of leaves given the remaining leaves, denoted by $i\dep j|\{1, \dots, Q\}\setminus\{i,j\}$.
\begin{figure}[H]
\centering
\includegraphics[scale=.44]{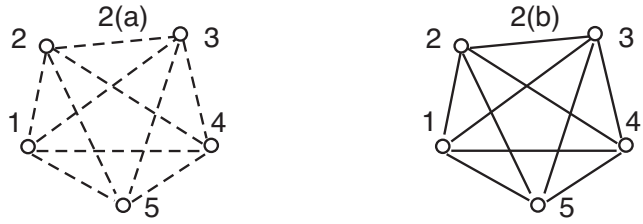}
\caption{Induced, complete graphs for the marginal distribution of the items, with the covariance graph in 2(a) resulting from 1(a) and the concentration graph in 2(b) from 1(b).}\label{fig:2}
\end{figure}

In exponential family distributions, a covariance graph corresponds to bivariate central moments, a concentration graph to joint canonical parameters and
a directed star graph to regression parameters.  If Markov equivalent models are also   parameter equivalent, then their parameter sets are in a one-to-one relation. This property does not  hold in general, but for instance in joint Gaussian distribution, and in joint  binary distributions  that are quadratic exponential.  More generally, 
Markov equivalence and parameter equivalence coincide when only a single parameter determines for each variable pair whether the pair 
is conditionally dependent or independent. 

For traceability of  a given sequence of regressions, the generated family of distributions  needs to have three properties of joint Gaussian distributions, that are  not common to all probability distributions. In addition to the dependence-inducing property (singleton-transitivity) stated in Definition 1, these are the {\bf intersection} (downward combination) and the {\bf composition}
property (upward combination of independences); see Ln\v{e}ni\v{c}ka and Mat\'u\v{s} (2007) equations (9) to (10)  and Wermuth (2012) section 2.4. 

Pairwise {\bf independences are said to  combine
downwards and upwards} if
$$(i\ci j|kc \text{ and } i\ci k|jc) \iff i\ci jk|c \nn \text{ and }\nn (i\ci j|c \text{ and }i \ci k|c) \iff i\ci jk|c\,,$$
respectively, for $c$ any subset of the remaining nodes, of $\{1,\dots Q\}\setminus \{i,j,k\}$. 
Both of these properties are already needed if one is using graphs of mixed edges  just to decide  on implied 
independences; see Sadeghi and Lauritzen (2014).

In the context of directed star graphs,  these two properties are  a consequence of the special type of generating process.
After removing any two arrows  for $Q\geq5$, a directed star graph in $Q-2$ arrows remains.  Thus, by introducing  two additional pairwise independences in a directed star graph, these independence statement combine downwards.
 The Markov equivalence of the directed to the undirected star graph implies that for any two nodes, the statement $i\ci j|L$ can be modified to  $i\ci j|Lc$ with $c \subseteq \{1,\dots Q\}\setminus \{i,j\}$ so that pairwise independences combine  upwards.

 \section{Gaussian distributions generated over directed star graphs}

In Gaussian distributions, each dependence is by definition linear and proportional to some (partial) correlation coefficient. Furthermore,   there are no higher than two-factor interactions.   In a traceable Gaussian distribution generated over  a directed star graph, each  directed edge  indicates a simple, strong  correlation, $\rho_{iL}$, 
called the {\bf loading of item $\bm i$ on $\bm L$}. Then for each item pair $(i,j)$ a simple correlation, $\rho_{ij}$ is induced   via 
\begin{equation} i\ci j|L\iff (\rho_{ij|L}=0) \iff  (\rho_{ij}=\rho_{iL}\rho_{jL}), \label{cindG}\end{equation}
where $\rho_{ij|L}$ denotes the partial correlation coefficient $$ \rho_{ij|L}=(\rho_{ij}-\rho_{iL}\rho_{jL})/\sqrt{(1-\rho_{iL}^2)(1-\rho_{jL}^2)}.$$

In factor analysis,   the latent root $L$ is assumed, without loss of generality,  to have mean zero and unit variance. Typically, when the  
observed variances of the items are nearly  equal,  the items  are transformed to have  mean zero and unit variance, so that their  correlation matrix is
analyzed. 

In general, the model parameters are known to be identifiable for $Q>2$; see for instance Stanghellini (1997).
For $Q=3$ items, the positive  loadings $\rho_{1L}, \rho_{2L}, \rho_{3L}$ can be completely recovered using equations \eqref{cindG} for the positive,  partial correlations $\rho_{ij | k}$  of each leaf pair. The  maximum-likelihood equations (Lawley, 1967) reduce to the same type of three equations so that	a unique,  closed form solution can be obtained, provided  it exists, and be written as $\hat{\bm \lambda}\T=(\hat{\rho}_{1L} \n \hat{\rho}_{2L} \n  \hat{\rho}_{3L})$ with
  \begin{equation} \hat{\rho}_{1L}= \sqrt{r_{12}r_{13}/r_{23}}, \nn \n  \hat{\rho}_{2L}= \sqrt{r_{12}r_{23}/r_{13}}, \nn \n  \hat{\rho}_{3L}= \sqrt{r_{13}r_{23}/r_{12}}.\label{estload}
 \end{equation}
 
  Clearly, these equations require for permissible estimates: $0< \hat{\rho}_{iL}<1$.  In that case, there can be no zero and no negative marginal or partial
  item correlation that cannot be removed by recoding some  items.
We give in Table~\ref{tab:heywood} three examples of $3\times 3$ invertible correlation matrices,  showing  marginal correlations  in the lower half  and  partial correlations   in the upper half. The first two examples have no permissible solution for $\hat{\bm \lambda}$ and illustrate so-called  Heywood cases.
  
\begin{table}[H]
\caption{Examples of invertible correlation matrices, on the left and in the middle: two Heywood cases that is no permissible solution to the maximum-likelihood equations; on the right: a perfect fit. }\label{tab:heywood}
\centering
$               \begin{matrix} \hline 
                           1 & 0.35 &  0.47 \\
                           0.40 & 1& 0 \\ 
                            0.50   &     0.20    &      1\\  \hline
                            \end{matrix}  \nn \fourl 
                            \begin{matrix} \hline 
                            1 & 0.64 &\n \; \,  0.55 \\
                           0.60 & 1& -0.29 \\ 
                            0.50   &     0.10    &      1\\ \hline
                         \end{matrix}, \nn \nn \nn
                         \begin{matrix} \hline 
                         1 & 0.57 &  0.39 \\
                           0.72 & 1& 0.20 \\ 
                            0.63   &     0.56    &      1\\\hline
                         \end{matrix}.              
                        $   
\end{table}                   
 In the example on the left,  $\rho_{23 | 1}=0$. By equation \eqref{estload}, the  estimated  loading of item 1 is then equal to one, i.e. $\hat{\rho}_{1L}=1$,  so that item 1 cannot be distinguished from the latent variable itself. For the example in the middle, $\hat{\rho}_{1L}=\sqrt{0.6\cdot 0.5/0.1}=\sqrt{3}$
 is larger than one, hence  leads to an infeasible solution for a correlation coefficient. Thus even for a  positive, invertible  item correlation matrix,  there may not exist a generating process via positive loadings on a single latent variable. The  example on the right, has a perfect fit for the  vector of estimated positive loadings  in equation \eqref{estload} as  $\hat{\bm{\lambda}}\T=[0.9 \nn  0.8 \nn 0.7]$.

For $Q>3$ items,   {\bf  proper positive loadings} in equation \eqref{cindG}, that is $0<\rho_{iL}<1$ for all items $i$, lead directly  to a positive
 correlation matrix  of the items, denoted in this paper by $\bm{P}>0$, that is to  exclusively positive correlations, $\rho_{ij}=\rho_{iL}\rho_{jL}$ for $i\neq j$ in $\{1,\dots,Q\}$, and to a tetrad structure.  Vanishing tetrads were defined  by Spearman and coauthors and nowadays a  popular search algorithm for models with possibly several latent variables  is named Tetrad. This algorithm is based on and extends work by Spirtes, Glymour and Scheines (1993, 2001).

\begin{defn}
{\bf A positive  tetrad correlation  matrix} $\bm{S}$ has dimension $Q >3$ and elements $s_{ij}$ such that $0<s_{ij}<1$ for all item pairs $i,j$ and
\begin{equation} s_{ih}/ s_{jh}=s_{ik}/ s_{jk}  \text{ for all distinct } i,j,h,k  \text{ taken  from } \{1, \ldots, Q\}.
\label{tetrad}\end{equation} \end{defn}
Thus for any  pair $i,j$,  the   ratio of its correlations to variables in  row $h$ of  $\bm S$ is the same as to variables in  row $k$, or, equivalently, 
there are {\bf vanishing tetrads}: $s_{ih}s_{jk} - s_{jh}s_{ik}=0$.
Let $\bm{\lambda}\T$ denote a row vector of proper positive loadings,  $\bm \Delta$ a diagonal matrix of elements $1-\rho_{iL}^2$,  and  $\bm{\delta}\T$ a row vector of elements $-\rho_{iL}/\{\sqrt{s}(1-\rho_{iL}^2)\}$,
 where  $s= 1+\sum_i \rho_{iL}^2/ (1-\rho_{iL}^2)$,  then, as proven in the Appendix, the correlation matrix of the leaves $\bm P$ and its inverse, the concentration matrix $P^{-1}$, are
  \begin{equation} \bm{P}=\bm{\Delta}+\bm{\lambda}\bm{\lambda}\T, \nn   \nn   {\bm P}^{-1}={\bm \Delta^{-1}}-{\bm \delta}{\bm \delta}\T \label{tetradcor}.\end{equation}
  
  Some important direct consequences of \eqref{tetradcor} are given next. Lemma 1 uses the notion of  M(inkowski)-matrices that were named and studied by Ostrowski (1937, 1956) and discussed in connection to totally positive Gaussian distributions much later using the name {\bf MTP$_2$ distributions}; see e.g. Karlin and Rinott  (1983). {\bf General MTP$_2$ distributions}  are characterized by having no variable pair  negatively associated given all remaining variables. We will use a more strict form of MTP$_2$ that also excludes any variable pair being conditionally independent given all remaining variables.

\begin{defn} 
A square, invertible matrix  is a {\bf complete  M-matrix} if all its diagonal elements  are positive and  all its off-diagonal elements are negative.
\end{defn}

\begin{lem} {\bf A positive tetrad correlation matrix, ${\bm P}$, generated over a star graph} with proper positive loadings, $0<\rho_{iL}<1$, is
invertible and ${\bm P}^{-1}$ is  a symmetric,  complete M-matrix with vanishing tetrads.
\end{lem}

When we denote the elements of  $\bm{P}^{-1}$  by $\rho^{ij}$, then by Lemma 1 and Definition 3,  all the precisions, $\rho^{ii},$ are positive and  all concentrations  are negative, that is $\rho^{ij}<0$, for all $i\neq j$, if $\bm{P}^{-1}$ is a complete M-matrix.
The following important  properties of complete M-matrices result from  Ostrowski (1956), Section 1.

\begin{lem} A symmetric matrix $\bm m$, which has an inverse complete M-matrix, \\
$(i)$ has exclusively positive elements, that is $\bm{m}>0$, and\\
$(ii)$  the {\bf inverse of every  principal submatrix of}  $\bm m$ is a complete M-matrix.
\end{lem}

Thus, if the concentration matrix $P^{-1}$ of observed Gaussian items  has exclusively negative off-diagonal elements, then 
all concentrations in every marginal distribution of the items are negative as well and ${\bm P}>0$.

\begin{lem}  {\bf Partial correlations, $\bm{\rho_{ij.c}}$ of a positive tetrad correlation matrix}, $\bm P$, generated over a star graph with proper positive loadings,  are positive for every $c\subseteq \{1, \ldots, Q\}\setminus\{i,j\}$ and form a positive tetrad matrix.
\end{lem}
\begin{proof} The result follows with  
the following general relation of elements of a concentration matrix to  partial correlations, together with Lemma 1 and Lemma 2.
 \begin{equation} \rho_{ij|c}=-\rho^{ij}/\sqrt{\rho^{ii}\rho^{jj}} \text{ for } c=\{1, \ldots, Q\}\setminus\{i,j\},  \label{conc}\end{equation} see for instance Wermuth, Cox and Marchetti (2006), Section 2.3. \end{proof}

 Thus, for a complete concentration matrix of Gaussian items, negative concentrations mean positive dependence and
 every item pair is positively dependent, no matter which conditioning set is chosen.
 Now, a known rank-one condition for the existence of a Gaussian item correlation matrix $\bm{P}^*\!$ of a single  factor simplifies as follows.

\begin{prop} {\bf Equivalent necessary and sufficient conditions for a Gaussian distribution}. For  $Q> 3$, the following statements are equivalent:\\ 
$(i)$  a  ${\bm P}^*>0$ is generated over a star graph  and with proper positive loadings, $\rho^*_{iL}$,\\
$(ii)$ a $\bm{P}^*>0$  minus a diagonal matrix, of elements $\delta_{ii}^*$ with $0<\delta_{ii}^*<1$,  has rank one,\\
$(iii)$ a  tetrad $\bm{P}^*>0$ of the items is formed by proper positive loadings, \\
$(iv)$  there is  a concentration M-matrix  $({\bm P}^*)^{-1}$ of the items having  vanishing tetrads,\\
$(v)$ there are item partial correlations, given the $Q-2$ other items,  which form a positive tetrad correlation matrix.
\end{prop}
\begin{proof}  In their terminology,   Anderson and Rubin (1956) proved $(i) \iff (ii)$ in their Theorem 4.1, without requiring proper positive loadings. They show also   that the rank-one  condition is equivalent to vanishing  tetrads provided $0\leq \rho_{ik}\rho_{jk}\leq \rho_{ij}$, that is for  ${\bm P}^*>0$ to $0<\rho_{ij.k} < 1$, for all distinct items $i,j,k$. For $\bm{P}^*>0$, this implies $(iii)$. By Lemma 3 and Lemma 1 above,
$(iii)  \iff (iv)$, by equation
\eqref{conc}, $(iv) \iff (v)$ and  equation \eqref{tetradcor} gives $(iii) \implies (ii)$.
\end{proof}

Conditions $(i)$ to $(iii)$ in Proposition 1 involve dependences of the items on the latent variable $L$, that is the unobserved loadings. 
In some applications, prior knowledge may be so strong that a positive loading can be safely predicted for each selected item, but otherwise,
these characterizations involve unknown parameters.

By contrast, conditions $(iv)$ and $(v)$ in Proposition 1 concern directly the distribution of only the observed items. Equivalence to the former conditions
for existence become possible by the added special properties of the concentration matrix $({\bm P}^*)^{-1}$ being a complete M-matrix with vanishing tetrads.

The following example shows that a positive, tetrad correlation matrix alone, does not assure that its inverse is  a M-matrix.
The example is another Heywood case,    conditions $(iv)$ and $(v)$ of Proposition 1 are violated:
\begin{equation}  {\bm P} =\begin{pmatrix} 
    1.00     &     0.84    &      0.84      &    0.84\\
          .       &   1.00      &    0.64      &    0.64\\
          . &          .    &      1.00     &     0.64\\
          . &          .&         .&        1.00\end{pmatrix}, \nn \fourl
{\bm P}^{-1}=\begin{pmatrix}
         16.20    &     -6.00    &     -6.00   &      -6.00\\
          .    &   \n \,4.22    &   \n \,   1.44    &   \n \,   1.44\\
         .& .            &   \n \,  4.22   &   \n \,    1.44\\
         .     &     .      & .  &     \n \,     4.22\end{pmatrix},
         \label{heywood}
 \end{equation}   
 where the dot-notation indicates symmetric entries.   
    
This correlation matrix cannot have been  generated over a directed star graph with  proper positive loadings.   If these correlations were observed, one would  get with equation $\eqref{estload}$ that $\hat{\rho}_{1L}>1$, that is not a permissible solution. By contrast, every ${\bm P}^{-1}$ that is a complete M-matrix with vanishing tetrads implies a positive  tetrad correlation matrix.
\section{Binary distributions generated over directed star graphs}

When $Q$ binary items are mutually independent given a binary  variable $L$ and the factorization of  the joint probability in equation \eqref{factden} cannot be further simplified, since each item has a strong dependence on $L$, then it is a traceable regression generated  over a directed star graph.
The reason is that binary distributions are, just like
Gaussian distributions via equation \eqref{cindG}, {\bf dependence-inducing}, that is 
$$ (i\dep L \text{ and }j\dep L)  \implies  \text{ at most } i\ci j|L  \text{ or } i\ci j \text{ but never both.} $$
The property assures for star graphs with   the independences of equation \eqref{cind}, $i \ci j|L$,   that $i\dep j$ is implied for each item pair. 
In applications,  strong dependences of each item $i$ on $L$ are needed to obtain relevant dependences for each leaf pair $i,j$.

The joint binary distribution  generated over a directed star graph is quadratic exponential since the $Q$ largest cliques in this type of a decomposable graph contain just two nodes, an item and  $L$. 
Expressed equivalently, in the  log-linear model of an undirected star graph, as in  Figure \ref{fig:1}(b), the  largest, non-vanishing  log-linear interactions are  positive  2-factor terms, $\alpha_{iL}$.  These are  canonical parameters in the generated binary quadratic exponential family.
Marginalizing over
$L$ in such distributions with $\alpha_{iL}>0$  gives a  tetrad form for the canonical parameters in the observed  item distribution; see Cox and Wermuth (1994), Section 3:
\begin{equation} \alpha_{ih}/ \alpha_{jh}=\alpha_{ik}/ \alpha_{jk}  \text{ for all distinct } i,j,h,k  \text{ taken  from } \{1, \ldots, Q\}. \label{tetradodd}\end{equation}
With $\alpha_{iL}>0$ for all $i$, the joint binary distributions of the items have exclusively positive  dependences.
In general, the parameters $\alpha_{iL}$
are identifiable for $Q\geq 3$, see Stanghellini and Vantaggi (2013). But equation  \eqref{tetradodd}  does not lead to a
an explicit form of the  induced bivariate dependences for the item pairs. 

For this, we write for instance for any two items $A$, $B$, and $L$, each with levels $0$ or $1$, 
$$ \pi_{ijl}=\Pr(A=i, B=j, L=l), \n \pi_{+j+}=\txt{\sum_{il}} \pi_{ijl}=\Pr(B=j), \n \pi_{ij+}=\Pr(A=i, B=j), $$
as well as
$$ \bm{\pi}_{AL} =\begin{pmatrix} \pi_{0+0} & \pi_{0+1}\\ \pi_{1+0} & \pi_{1+1}
 \end{pmatrix},   \nn \nn  \bm{\pi}_{LB} =\begin{pmatrix} \pi_{+00} & \pi_{+10} \\
\pi_{+01}  &\pi_{+11} 
 \end{pmatrix},  \nn \nn  \bm{\pi}_{L} =\begin{pmatrix} \pi_{++0} & 0\\
0& \pi_{++1} 
 \end{pmatrix}.$$
 
 There are  three equivalent expressions of Pearson's correlation coefficient for the binary items 1 and 2. To present these, we  use $\kappa_{12}=\sqrt{\pi_{0++}\pi_{1++}
 \pi_{+0+}\pi_{+1+}}$ and abbreviate  the operation of taking a determinant by $\det(.)$\,:
  \begin{eqnarray} \label{defrho}
  \rho_{12}&=&(\pi_{00+}\pi_{11+}-\pi_{10+}\pi_{01+})/\kappa_{12} \label{odr_rho}\, ,\\[1mm]
              &=&(\pi_{11+}-\pi_{1++}\pi_{+1+})/\kappa_{12} \label{cov_rho} \, ,\\[2mm]
               &=& \det(\bm{\pi}_{A}^{-1/2} \bm{\pi}_{AB}\;\bm{\pi}_{B}^{-1/2}) \label{det_rho}\, .
                 \end{eqnarray} 
               
 Equation \eqref{odr_rho} shows that the correlation is given by the cross-product difference, that it  is zero if and only if the odds-ratio, defined as the cross-product ratio, equals one and that a positive correlation is equivalent to a positive log-odds ratio. Equation \eqref{cov_rho}  gives as numerator 
 the covariance and as denominator the product of two standard deviations.  This is the usual definition of Pearson's correlation coefficient, here  for binary variables, and it implies that $\rho_{12}$
 together with the one-dimensional frequencies of items 1 and 2 give the counts in their $2\times 2$ table.  Equation \eqref{det_rho} leads
best to our main new result.

 From equations \eqref{factden} and  \eqref{det_rho}, we have for each {\bf source {\sf V}} of the star graph, that is for each configuration $A \fla L \fra B$, a trivariate binary distribution with $A\ci B|L$ and
  \begin{equation}  \bm{\pi}_{AB}=\bm{\pi}_{AL}(\bm{\pi}_L)^{-1} \bm{\pi}_{LB}\,, \\ \label{cind_rho}\,.\end{equation}
 Hence, the existence of a correlation matrix $\bm{P}^*$ becomes relevant for binary variables.
 \begin{prop} \!\!{\bf Equivalent necessary conditions for
joint binary distributions.} For $Q>3$ items to be generated over a directed star graph with a  binary root $L$
\\
$(i)$ there is a  tetrad correlation matrix $\bm{P}^*>0$ formed by proper positive loadings, $\rho_{iL}$, \\
$(ii)$ there are item partial correlations, given the $Q-2$ other items, which  form a positive tetrad correlation matrix.
\end{prop}
\begin{proof} Premultiplying equation \eqref{cind_rho} by $\bm{\pi}_A^{-1/2}$ and post-multiplying it by $\bm{\pi}_B^{-1/2}$ gives, for $0<\rho_{iL}<1$  for all $i$, with equation \eqref{det_rho}, that is  after taking determinants, $ \rho_{12}=\rho_{1L}\rho_{2L}>0\, .$
Since this holds for all item pairs,
\begin{equation}       \rho_{ij}=\rho_{iL}\rho_{jL}>0 \text{ for all } i,j \in \{1, \ldots, Q\},
\end{equation}
and the positive tetrad correlation matrix is a consequence
of the generating process.  The same arguments as in Proposition 1 give the equivalence of the two statements.
\end{proof}
This result corrects a claim in Cox and Wermuth (2002) that a tetrad condition does not show in correlations of binary items but only in their  canonical parameters  in the induced distribution for  the $Q$ items:
just as in Gaussian distributions, a  positive, invertible tetrad  correlation matrix, $\bm P$, 
is induced for binary  items if the dependence of each item on the latent binary $L$  is  positive, that is if $\rho_{iL}>0$ for all $i$.

 Nevertheless, the  directly relevant  dependences for  the induced, complete concentration graph model are the canonical parameters, the log-linear interaction terms that are functions of the odds-ratios. For instance, for a binary distribution  generated over a star graph to have  a general MTP$_2$ distribution, the condition is $\alpha_{iL} \geq 0$,  while  for a  strictly positive subclass, the condition is $\alpha_{iL} > 0$ for all items $i$.

Necessary and sufficient conditions, that a density of equation \eqref{factden}  has generated the observed distribution of $Q=3$ items, have been derived as nine inequality constraints on  the probabilities of the item by Allman et al. (2013) without noting their relation to MTP$_2$ distributions.
For $\pi_{ijk} >0$, the first three reduce to
$$ \pi_{111+}/  \pi_{011+}   \geq  \pi_{100+}/\pi_{000+},\nn
	\pi_{111+}/\pi_{101+}   \geq  \pi_{010+}/  \pi_{000+}  ,\nn
	\pi_{111+}/ \pi_{110+} \geq  \pi_{001+}/  \pi_{000+} ,\label{zwiern} $$
so that for each leaf,  the odds for level 1 to 0  when the levels of the other two leaves  match at level 0 do not exceed those with matches at level 1. Their last six constraints  require nonnegative odds-ratios for each leaf pair, that is a binary MTP$_2$ distribution.

It can be shown that the above three inequalities  are satisfied, whenever each leaf pair $i,j$ has a  positive 
conditional dependence given the $Q-2$ remaining leaves, that is if the leaves have a strictly positive distribution. Therefore, this strict form of a MTP$_2$ binary distribution of the leaves is also sufficient for just $Q=3$ leaves to have been generated over a star graph with positive dependences of the leaves on the root. This implies but is not equivalent to ${\bm P}^{-1} $ having exclusively negative off-diagonal elements.

More complex characterizing inequality constraints on  probabilities  of the leaves for $Q>3$ categorical variables,  when leaves and root have the same number of levels,
are due to Zwiernik and Smith (2011). It remains to be seen how  they simplify for  binary  MTP$_2$  distributions of the leaves  or with a complete, tetrad   ${\bm P}^{-1} $. 

\section{Applications}

We use here three sets of binary items. The  first is a medical data set, the last two are psychometric data sets where the questions were  chosen 
to expect strong positive dependences of each item on a latent variable. 
As discussed above, to check   conditions for the existence of a single latent variable, that might have generated the observed item dependences,
we use  here mainly the observed item correlation matrices and the observed  marginal  tables of all item triples.
In the first two cases, no violations are detected. In the third case,  item correlation matrices alone provide already enough evidence against
the hypothesized generating process. Algorithms to compute maximum-likelihood estimates for latent class models, are widely available; see for instance
Linzer and Lewis (2011).

\subsection{Binary items indicating gestosis that may arise during  pregnancy}
Worldwide, EPH-gestosis is still the main cause for a woman's death during childbirth and a major risk for death of the child during birth or within a week after birth. It is until today not a well-understood  illness, rather it is characterized by the occurrence of two or more symptoms, of  edema (E:=high body water retention), proteinuria (P:=high amounts of urinary proteins) and  hypertension (H:=elevated blood pressure).

Little research  into  causes of EPH-gestosis appears to have  been undertaken during  the last 50 years, possibly  because  in higher developed countries  its worst negative consequences are avoided 
by intervening when two of the symptoms are observed. Our data are from the prospective study `Pregnancy and Child Development'  in Germany, see the research report of the German Research Foundation (DFG, 1977). The symptoms were  recorded before
birth for 4649 pregnant women. 

As a convention, we order in this paper  counts reported in vectors such that the levels change from 0 to 1 and the levels of the first variable changes fastest, those of the second change next and so on. The observed counts for the gestosis data are then
$$ \bm{n}\T =( 3299 \nn 78  \nn   107  \nn  11 \nn 1012   \nn  65  \nn  58  \nn   19).
$$
 There are exclusively positive conditional dependences for  the symptom pairs since all odds-ratios are larger than 1; with values  
 of 4.4 and 5.1  for E,P, of 2.7 and 3.2 for E,H and of 1.8 and 2.1 for P,H, where the given level of the third symptom changes from 0 (absence) to 1.
 Except for P,H at level 1 of E, the corresponding confidence intervals exclude negative dependences. The observed relative frequencies satisfy the relations of equation \eqref{zwiern} directly.
 Thus, the above observed counts support the hypothesis of a generating directed star graph with a latent binary root.

Some additional features of the data are given next. The first symptom (E) is present for 3.7\%, the second (P) for 4.2\% and the third (H) for 24.8\% of the women.  Of the symptom pairs, 
E,P are seen for 65, E,H for 181 and P,H for 166 of 1000 women and  all three symptoms  for 41 of 1000 women. The bivariate dependences are strong and positive; with values for the odds-ratios of 5.5 for E,P, 3.0 for E,H and 2.0 for P,H. 

The corresponding correlations
look smaller than expected for quantitative features because E and P are rare symptoms. They have values  0.13 for E,P,  0.11 for E,H,  and 0.07
for P,H and the inverse of the observed correlation matrix is a complete M-matrix.

\begin{table}[H]
\caption{lower triangle: marginal item correlations and loadings,  upper triangle:  partial correlations.}\label{tab:eph}
\centering
$
\begin{array}{ccc}
\begin{matrix} \hline
  1 & 0 &  0 & 0.42\\
  0.13 & 1&  0& 0.26\\ 
   0.11  &     0.07    &      1&0.21 \\ 
   0.44&0.29& 0.24& 1 \\ \hline
 \end{matrix}    \end{array}                       
$
\end{table}

By equating the standardized central moments to the observed correlations, the estimates $\tilde{\rho}_{iL}$ are added in the last row to the observed correlations in  
Table \ref{tab:eph}; used is   the order E,P,H,$L$.
The identity matrix within the  matrix of partial correlations given the remaining two variables indicates the perfect fit of the correlations to the hypothesized generating process via a star graph.

\subsection{Binary items in a  small depression scale}

From an evaluation study   of a short depression scale developed by 
 Hardt  (2008), 
we use four binary items.  The answers  (no:=0, yes:=1)  are to the questions:  feeling    
hopeless (item 1, with 35.1\% yes),  dispirited (item 2, with 27.6\% yes), empty inside (item 3, with 24.2\% yes),  loss of  happiness (item 4, with 33.0\% yes). 
The  observed row vector of counts is, again ordered as described in section 5.1,
$$ \bm{n}\T=(533\nn     52   \nn  22  \nn   27  \nn    8  \nn   15  \nn    4   \nn  14 \nn 46  \nn   32  \nn    4  \nn   48   \nn  19  \nn   25 \nn    18   \nn 141).$$
All conditional odds-ratios for items 1 and 2 are positive, with values  12.6,    1.9,   17.3,   6.0, respectively  for levels  (00, 10, 01,11) of items 3 and 4.
Even though there are 1008 respondents,  some  subtables contain only  small numbers.    
Especially for many items, tests in such tables have little power and may therefore not be very informative. 
We therefore concentrate on trivariate subtables.

For each triple of the items, we show  in Table~\ref{tab:loglin}  studentized log-linear interaction parameters for the 2-factor terms and the 3-factor term. Each is a log-linear term estimated under the hypothesis that it is zero and 
divided by an estimate of its standard deviation; see for instance  Andersen and Skovgaard (2010). To simplify the display, we list the involved item numbers,
but use the same notation for interaction terms, for instance $AB$ is for the first two listed items, for items 1,2 in the triple 1,2,3 but for items 2,3 in the triple 2,3,4.
\begin{table}[H]
\caption{Studentized interaction parameters. For each item triple, $A:=$  first, $B:=$   second $C:=$ third.}\label{tab:loglin}
\centering
$     
\begin{matrix}
\hline
\text{Item triple}&AB& AC&BC&ABC \\ \hline
1,2,3&10.7   &     7.6  &       10.3& 3.4 \\ 
1,2,4& 11.9   &    9.0  &        \n9.0 &   0.6\\
1,3,4&     8.7 &       6.8  &       11.8&       4.1\\
2,3,4&           9.2 &     8.4 &        11.2 &1.3\\ \hline\\[-3mm]
\end{matrix} 
$ 
\end{table}

         The 3-factor interactions for items
                            1,2,3 and for 1,3,4 are not negligible but all interaction terms are positive and the 3-factor term
                            is always smaller than any  of the 2-factor terms.
                                                        
                            The similarity of the positive dependences at each level of the third variable shows here best in the  two conditional  relative
         risks for the first pair in each item triple: 
         $$ (5.9,\nn 1.5;\nn \nn  5.6, \nn 1.9; \nn \nn 5.7,\nn 1.3;\nn \nn  5.7,\nn 2.0).      
$$

                            The marginal observed correlations, are shown  next   with Table~\ref{tab:1}(a),  in the lower triangle,  and the partial correlations computed with equation \eqref{conc} in the upper triangle. The same type of display is used in Table~\ref{tab:1}(b) with an estimated vector $\tilde{\lambda}\T$ of  loadings, $\tilde{\rho}_{iL}$,  added in the last row, computed as if the correlation matrix were a sample from a Gaussian distribution.                             
                            \begin{table}[H]
\caption{(a) Marginal item correlations in the lower triangle, partial correlations given two remaining items in the upper triangle;
 (b)  lower triangle: as in Table~\ref{tab:1}(a) but with a  row vector of loadings added,  upper triangle: corresponding partial correlations given the remaing three variables.}\label{tab:1}
\centering
$
\begin{array}{ccc}
\begin{matrix} \hline
  1 & 0.37 &  0.16&  0.26\\
  0.62 & 1&  0.26& 0.20\\ 
   0.53  &     0.57    &      1&0.36 \\ 
   0.57&0.56& 0.60& 1 \\ \hline
 \end{matrix}                           
&\qquad  \fourl \fourl &
\begin{matrix} \hline
        1 & 0.08 &        -0.07   &       \n \;0.00   &       0.45\\
        0.62 & 1&  \n \;0.00  &       -0.06    &      0.48 \\ 
        0.53  &     0.57    &       \n \;1&     \n \;0.08  &        0.44 \\ 
        0.57&0.56&  \n \;0.60&  \n \;1 & 0.46 \\
        0.76 &         0.77    &       \n \;0.74    & \n \;0.76  &     1\\ \hline
\end{matrix}
\\
(a) & & (b) \\
\end{array}
$ 
\end{table}
All marginal correlations are positive and strong in the context of binary answers to questions concerning feelings. All partial correlation given the remaining two items are also positive. Thus, these two necessary conditions for the existence of a single latent variable are satisfied by the given counts.
With  $\tilde{\rho}_{iL}$ added in Table~\ref{tab:1}(b), the partial correlations for each item pair, given the remaining two items and the latent 
variable, are  quite close to zero, as they should be under the model.

\subsection{Binary items in a  failed attempt to construct a scale}
In the same study of the previous section, the participants were asked whether  their parents  fought often. For 538  respondents who answered  yes, 
answers to reasons of the fighting were  (no:=0, yes:=1) to:     
hot temper (item 1, with 50.0\% yes),  money  (item 2, with 46.5\% yes), alcohol (item 3, with 36.1\% yes),  jealousy (item 4, with 24.3\% yes). 
The observed vector of counts is, again ordered as described in section 5.1, 
$$ 
\bm{n}\T=(46\nn     113   \nn  63  \nn   55  \nn   59  \nn   14  \nn    42   \nn  15 \nn 10  \nn   23  \nn    17  \nn   17  \nn  13  \nn   10\nn    19   \nn 22).
$$

Tables~\ref{tab:2}(a)  and \ref{tab:2}(b) contain marginal correlations in the lower half of a  matrix; Table~\ref{tab:2}(a)  for all four items and Table~\ref{tab:2}(b)  for items 1 to 3. In the upper half  are  the partial correlations, computed with equation \eqref{conc} from the corresponding overall concentration matrices.

\begin{table}[H]
\caption{(a) Marginal (lower half) and partial item correlations (upper half) given the two remaining  items; (b) For items 2,3,4, marginal (lower half) and partial correlations (upper half)} \label{tab:2}
\centering
$
\begin{array}{ccc}
\begin{matrix}\hline
   1 & -0.12 &  -0.29&0.12\\
   -0.12 &\n  1& \n \; 0.01& 0.13\\ 
    -0.28  &  \n \,   0.06    &      1&0.17 \\ 
    \n \; 0.06&  \n\; 0.12& \n \, 0.15&  1 \\ \hline
\end{matrix} 
& \qquad \fourl \fourl&                          
\begin{matrix} \hline
   1 & 0.04 &      0.11   \\
   0.06 & 1& 0.15  \\ 
    0.12  &     0.15    &    1\\ \hline
\end{matrix} \\
(a) & & (b) \\
\end{array}
$
\end{table}
In Table~\ref{tab:2}(a),  item 1 has negative dependences on items 2 and 3, the one on item 3 is strong ($\sqrt{538}(-0.28)=-6.5$). This correlation cannot  result even when $\rho_{3L}=0$, let alone when $\rho_{3L}>0$.
For the three remaining items, the correlations in Table~\ref{tab:2}(b) are all positive, but too small to support the hypothesized generating process with a  useful  indicator variable  $L$.

{\section{Discussion}} 

\subsection{Factor analysis applied to data with concentration M-matrices}
After Spearman had introduced factor analysis in 1904, he and others
claimed that the vanishing of the tetrads is necessary and 
sufficient for the existence of a single latent factor. It
was Heywood (1931) who proved that the condition was only
necessary when negative correlations 
(marginal or partial) are also permitted. His results  implied 
that $0\leq \rho_{ij.k}\leq 1$
for all distinct item triples is needed, in addition to vanishing tetrads, for a
necessary and sufficient condition, in general. The same result 
was proven by Andersen and Rubin (1956) using a  
rank  condition.

Heywood mentions that under some sort
of strict positivity, a vanishing of tetrads may indeed give 
a necessary and sufficient condition, but the 
relevant properties of  a M-matrix  were unknown  
at the time. These relevant features, used here in Proposition~1 $(iv)$, were derived  by Ostrowski (1956) 
without having any applications  in statistics in mind. The connection 
of M-matrices to  Gaussian concentration matrices 
was only recognized much later by Bolviken (1982).
It may  also be checked that Spearman's  (1904) applications lead to complete  concentration M-matrices.

 For Gaussian distributions, complete concentration M-matrices define  a subclass  that is even more constrained than the one  that is MTP$_2$, where 
 off-diagonal zeros may arise and indicate  conditional independence;  see e.g. Karlin and Rinott (1983).

 Proposition 1  concerns  a generating process via a directed star graph for  a Gaussian correlation matrix.  Important 
are  the two equivalent constraints,  $(iv)$ and  $(v)$,  on only the observable  distribution: the concentration matrix  of the leaves is a complete M-matrix with vanishing tetrads and partial 
correlations of leaf-pairs $i,j$, given the remaining $Q-2$  leaves form a positive tetrad correlation matrix. The second of these two is easier to recognize due to the scaling of correlations.

\subsection{Applications of binary star graphs}

Joint binary distributions generated over general
 directed star graphs are extremely constrained; see Figure  1  in Allman 
 et al. (2013) and even more for strictly positive  conditional dependences of the leaves on the root, as these  define a special subclass of the general binary  MTP$_2$ family,  studied by  Bartolucci and Forcina (2000).

 Nevertheless, the  structure in the medical and in one psychometric data set of Section 5.3 can be explained by  such a generating process.  There,  strong prior  knowledge about $\rho_{iL}$, based on observing many patients, permits to select suitable sets of symptoms or items: these are  three symptoms for 
EPH-gestosis and four  suitable binary items that may indicate depression.

 Inequality constraints  on probabilities of the  joint  distribution of the leaves  have been given by Zwiernik and Smith (2011) in their Proposition 2.5, restated in simplified form  for $\bm{P}>0$ and $Q=3$ by Allman et al (2013).  The latter are discussed here using equation \eqref{zwiern}.
 Compared to  Gaussian distributions, Proposition 2 contains   the same conditions on the  correlations of the leaves and on the partial correlations of the leaves. The latter are easy to check but they are for general types of binary variables  only necessary conditions.

Applications of phylogenetic star graphs and trees, as started by Lake (1994), were based on incomplete
characterizations that have been completed only recently by  Zwiernik and Smith (2011). It is still unclear whether the  history of factor analysis may  repeat itself  in this context:
the current  applications are  plagued by infeasible solutions, but possibly 
there are unknown characterizing  features for  some situations, under which  these problems are always  avoided.

Though Lake's  `paralinear distance measure' reduces 
to $-\log(|r_{ij}|)$ for Gaussian and for binary variables, little is known about  the
distribution of the corresponding random variables. Even for a  Gaussian parent distribution, the variance in the asymptotic Gaussian distribution of 
$\hat{\rho}$ involves the unknown $\rho$ unless $	\rho=0$. To	 construct confidence intervals for $\rho \neq 0$, one uses Fisher's z-tranformation, which lacks this undesirable feature. However, for other than Gaussian parent distributions, even the z-transformed, correlation coefficient estimator  depends  in general  on
the unknown $\rho$;
see Hawkins (1989). 

Sometimes,  as  for the data in Section 5.3 above, the absolute value of an observed  negative simple or partial correlation is so large that it clearly contradicts the existence of a generating
star graph with only proper positive dependences on the root.
 But, correlations alone or constraints on the population probabilities alone, such as in equations \eqref{zwiern}, cannot help to decide
 whether an observed negative dependence, $r_{ij}\leq 0$, may arise from a population in which  $\rho_{ij}>0$.
\subsection{Machine learning procedures for star graphs with a latent root}
In the  machine learning literature, it is considered to be one of the simpler tasks to decide whether a joint binary distribution
has been generated over a directed star graph. 

   However, when a learning strategy is based on only  the bivariate  binary distributions, no joint distribution may exist for a set of given bivariate distributions. In the spirit of Zentgraf's  (1975)
    example, we take  $2 \times 2$ tables of counts for variable pairs  $(A,B)$;  $(A,C)$; $(B,C)$; again with the levels of the first variable changing fastest:
    \begin{equation}\bm{n}\T_{AB}=( 77   \n\, 41    \n\,   101  \n\,   221), \n  \bm{n}\T_{AC}=(105  \n\,  41  \n\,73  \n
    \, 221), \n  \bm{n}\T_{BC}=( 45 \n\,101  \n\, 73 \n,          221).   \end{equation}
    The  odds-ratios  are 4.1,          7.8, and          1.4, respectively.
    These are marginal tables of the following $2\times 2\times 2$ table of counts    $$  \bm{n}\T=( 19  \nn           26    \nn         86   \nn          15  \nn 58    \nn         15     \nn        15   \nn         206  ).   
    $$
    The conditional odds-ratio of $A,B$ given $C$, at level 0 of 0.13 and given $C$ at level 1 of 53.1, show  qualitatively strongly different dependences 
    of $A, B$ given $C$. When  inference is based only on  the  bivariate tables, one implicitly sets the third-order 
    central moment to zero and keeps all others unchanged. Transforming this vector back to probabilities gives negative entries and hence 
    shows that no joint binary distribution exists when   the log-linear,  three-factor interaction  is falsely taken to be zero. 
        
\section{Appendix: The inverse of a positive tetrad correlation matrix}

It was known already to Bartlett (1951), that an invertible  tetrad correlation matrix  implies a tetrad concentration matrix. His proof is in terms of  the general form  of the inverse
of sums of matrices. It is more direct to   give the overall concentration matrix  in explicit form.

For this, we use the partial inversion operator  of Wermuth, Wiedenbeck and Cox (2006), described in the context of   Gaussian parameter matrices  in Marchetti and Wermuth (2009). It can be viewed as a Gaussian elimination technique (for some history  see Grcar, 2011) and as a minor extension of the sweep operator for symmetric matrices discussed by Dempster (1972). This extension is to invertible, square matrices  so that an operation is undone by just reapplying the operator to  the same set. 

Let $\bm M$ be a square matrix of dimension $d$ for which all  principal submatrices are invertible.   To describe partial inversions on $d$, we  partition $\bm M$ into  a matrix $\bm m$ of dimension $d-1$, column vector $\bm v$,  row vector $\bm{w}\T$ and scalar $s$ 
\begin{equation} {\bm M}=\begin{pmatrix} \bm m& \bm v\\ \bm{w\T}& s\end{pmatrix}\!, \nn  \n \inv_d\, \bm{M}=\begin{pmatrix} \bm{m}-\bm{vw}\T/s & \bm{v}/s\\ -{\bm w}\T/s& 1/s \end{pmatrix}. \label{pinvd}\end{equation} The transformation of  $\bm m$ is also known as the vector  form of  a Schur complement; see Schur (1917).
For $\bm M$ a  covariance matrix, ${\bm v}/s$ contains  $d-1$  linear, least-squares  regression coefficients and 
$\bm{m}-\bm{vw}\T/s$ is  a residual covariance matrix.

For partial inversion on a set  $a \subseteq \{1, \ldots, d\}$,  one may conceptually apply equation \eqref{pinvd} repeatedly for each index $k$ of $a$: one first reorders the matrix $\bm M$ so that $k$ is the last  row and column,
applies equation \eqref{pinvd} and returns to the original ordering.  Several useful and nice properties of this operator have been derived, such as commutativity and symmetric difference.

For $Q=3$ leaves and the  correlation matrices of a star graph models, 
more detail  is
$$ 
\bm{\Psi}=\left(\begin{array}{rrrr} 
                 1&  \rho_{1L}\rho_{2L}& \rho_{1L}\rho_{3L}&\rho_{1L}\\
                 . & 1& \rho_{2L}\rho_{3L}& \rho_{2L}\\
                 .&.& 1& \rho_{3L}\\
                 .&.&.& 1
                    \end{array} \right), \nn \nn \nn 
                 \inv_L  \,  \bm{\Psi}=\left(\begin{array}{rrrr} 
                 1-\rho_{1L}^2&  0& 0&\rho_{1L}\\
                 . & 1-\rho_{2L}^2& 0& \rho_{2L}\\
                 .&.& 1-\rho_{3L}^2& \rho_{3L}\\
                 \scriptstyle \sim \n&\scriptstyle \sim \n&\scriptstyle \sim \n& 1
                    \end{array} \right),
 \label{starmat} $$

\begin{equation} \bm{\Psi}^{-1}=\left(\begin{array}{rrrr} 
                 1/(1-\rho_{1L}^2)&  0& 0&-\rho_{1L}/ (1-\rho_{1L}^2)\\
                 . & 1/(1-\rho_{2L}^2)& 0& -\rho_{2L}/(1-\rho_{2L}^2)\\
                 .&.& 1/(1-\rho_{3L}^2)&- \rho_{3L}/(1-\rho_{3L}^2)\\
                 .&.&.& 1+\sum_i \rho_{iL}^2/ (1-\rho_{iL}^2)
                    \end{array} \right), \label{invcor}
 \end{equation}
  where the
                   $ .\n$-notation indicates an entry that is symmetric, the $\scriptstyle \sim \n$-notation an entry  that is symmetric up to the sign.

For $Q$ items, a single root $L$ and $\bm \Psi$ denoting their joint correlation matrix, we have for instance 
 $$ \inv_L  \,  \bm{\Psi}=  \inv_{1,\dots, Q}\, \bm{\Psi}^{-1}, \nn \nn  \bm{\Psi}^{-1}= \inv_{1,\dots, Q}\,(  \inv_L  \, \bm{\Psi}).$$
                 For $Q>3$, the structure of these matrices is preserved in the sense that there is a diagonal matrix $\bm \Delta$ containing  $1-\rho_{iL}^2$ as elements,    a row vector $\bm{\lambda}\T$ with loadings $\rho_{iL}$,   the precision of $L$ as $s= 1+\sum_i \rho_{iL}^2/ (1-\rho_{iL}^2)$, 
            and        a row vector  $\bm{\delta}\T$  with elements $-\rho_{iL}/\{\sqrt{s}(1-\rho_{iL}^2)\}$.

                   For  $\bm P$ the correlation matrix of  the $Q$ items that are uncorrelated given $L$, one gets $\bm P$ as the submatrix of rows and columns $\{1, \ldots, Q\}$
                    of
                     $$\inv_L( \inv_L  \, \bm{\Psi})$$ and  $\bm{P}^{-1}$  as the submatrix of rows and columns 
                    $\{1, \ldots, Q\}$ of $ \inv_L  \, \bm{\Psi}^{-1} $ so that

                   $$ \bm{P}=\bm{\Delta}+\bm{\lambda}\bm{\lambda}\T, \nn   \nn   \bm{P}^{-1}=\bm{\Delta}^{-1}-\bm{\delta}\bm{\delta}\T .$$
                   Thus, for $ \bm{P}^{-1}$ a complete M-matrix with vanishing tetrads,  $ \bm{P}>0$ has tetrad form.\\

\noindent{\bf Acknowledgement.}  
We thank  referees and colleagues, especially  D.R. Cox and P. Zwiernik, for their constructive comments and J. Hardt for  letting us analyze his data. We used MATLAB for the computations.


\begin{thebibliography}{99}

\bibitem{AllmanEtal13}
\textsc{Allman, E.S., Rhodes, J..A., Sturmfels, B. and Zwiernik, P.} (2013).
Tensors of nonnegative rank two.  ArXiv:1305.0539 and Lin. Algeb. \& Appl. To appear 2014.

 
\bibitem{AndSkov10}
\textsc{Andersen P.K. and  Skovgaard L.T. } (2010).
\textit{ Regression with linear predictors.}
 Springer, New York.

\bibitem{AndRub56}
\textsc{Anderson, T.W.  and Rubin H.} (1956). 
Statistical inference in factor analysis. In: \textit{Proceedings of the
3rd Berkeley Symposium on Mathematical Statistics and Probability}, \textbf{5},
University of California Press, Berkeley, 111--150. 


\bibitem{Bartl35}
\textsc{Bartlett, M.S.} (1935).
 Contingency table interactions.
 \textit{Supplem.  J. Roy. Statist. Soc.}
 \textbf{2}, 248-252.

 


\bibitem{Bartl51}
\textsc{Bartlett, M.S.} (1951).
 An inverse matrix adjustment arising in discriminant analysis.
\textit{Ann. Math. Statist.} \textbf{22}, 107--111.

 
  \bibitem{BartForc00}
\textsc{Bartolucci, F. and Forcina, A.} (2000).
A likelihood ratio test for MTP$_2$ within binary variables.
\textit{Ann. Statist.} \textbf{28}, 
 1206--1218.


  \bibitem{Bolv82}
\textsc{Bolviken, E.} (1982).
 Probability inequalities for the multivariate normal with non-negative
partial correlations, 
 \textit{Scand. J. Statist.}
 \textbf{9}, 49--58.



\bibitem{Cox72}
\textsc{Cox, D.R.} (1972).
The analysis of multivariate binary data.
\textit{J. Roy. Statist. Soc. C}
\textbf{21}, 113--120.


\bibitem{CoxWer94}
\textsc{Cox, D.R. and Wermuth, N.} (1994).
A note on the quadratic exponential binary distribution.
\textit{Biometrika} \textbf{81}, 403--406.

 \bibitem{CoxWer02} 
 \textsc{Cox, D.R. and  Wermuth, N.} (2002). On some models for multivariate binary variables parallel in
complexity with the multivariate Gaussian distribution.
\textit{Biometrika} \textbf{89}, 462--469.

 
 
\bibitem{Dem72}
\textsc{Dempster, A.P.} (1972).
Covariance selection.
\textit{Biometrics} \textbf{28}, 157--175.
	
\bibitem{DFG76}
\textsc{DFG-Forschungsbericht} (1977).
\textit{Schwangerschaftsverlauf und Kindesentwicklung.} Boldt,
Boppard. 

\bibitem{Edw63}
\textsc{Edwards, A. W. F.} (1963).
The measure of association in a 2 $\times$ 2 table.
\textit{J. Roy. Statist.  Soc.  A}
\textbf{126},
109--114.


  \bibitem{EsaryProschanWalkup67} 
 \textsc{Esary,  J.D., Proschan F. and  Walkup, D.W. } (1967).
   Association of random variables, with applications.
\textit{Ann. Math.  Statist.} \textbf{38}, 1466--1474.

  
  \bibitem{Grcar11}
  \textsc{Grcar, J. F.}  (2011). Mathematicians of Gaussian elimination.
  \textit{Notices Amer.  Math.  Soc.}
  \textbf{58},  782--792.

  
 \bibitem{Hardt08}
 \textsc{Hardt, J.} (2008).
 The symptom checklist-27-plus (SCL-27-plus): a modern
conceptualization of a traditional screening instrument.
 \textit{GMS Psycho-Soc-Medicine}  \textbf{5}.

 \bibitem{Hawk89}
\textsc{Hawkins, D.L.} (1989). Using U statistics to derive the asymptotic distribution of Fisher's z-statistic.
\textit{Amer. Statist.} \textbf{43}, 235--237.

\bibitem{Heyw31}
 \textsc{Heywood, H.B.} (1931). On finite sequences of real numbers. 
 \textit{Proc. Roy. Soc. London
A} \textbf{134}, 486--501.


 

 
\bibitem{HollRos86}
\textsc{Holland P.W. and Rosenbaum, P.R.} (1986).
Conditional association and unidimensionality in monotone latent
variable models.
\textit{Ann. Statist.} \textbf{14}, 1523--1543.

 
 \bibitem{KarlRin83}
  \textsc{Karlin, S.  and Rinott, Y} (1983). M-matrices as covariance matrices of multinormal
distributions. 
\textit{Linear Algebra Appl.}
\textbf{52}, 419--438.
 

 \bibitem{Lake94}
 \textsc{Lake A.} (1994).
Reconstructing evolutionary trees from DNA and protein
sequences: paralinear distances.
 \textit{Proc. National Acad. Sci. USA}
 \textbf{91}, 1455--1459.

 
\bibitem{Lawley67}
{Lawley, D.N.} (1967).  
Some new results in maximum likelihood factor analysis.
\textit{Proc. Roy. Soc. Edinb. A}
\textbf{67},  256--264. 
 
\bibitem{Lazar50}
\textsc{Lazarsfeld, P.F.} (1950). The logical and mathematical foundation of latent structure analysis.
\textit{Measurement Prediction} \textbf{4}, 362--412.


\bibitem{LinzLew11}
\textsc{Linzer, D.A. and Lewis J.B.} (2011).
{poLCA: An R package for polytomous variable latent class analysis}
\textit{J. Statist. Softw.}  \textbf{42}.




 \bibitem{MairHatz07}
\textsc{Mair, P.  and Hatzinger, R.} (2007)
Extended Rasch modeling: the eRm package for
the implication of IRT models in R.
\textit{J. Statist. Softw.},  \textbf{20, 9}


\bibitem{MarWer09}
\textsc{Marchetti, G.M. and Wermuth, N.} (2009).
Matrix representations and independencies in
directed acyclic graphs.
\textit{Ann. Statist.}
\textbf{47}, 961--978.


\bibitem{Neym71}
\textsc{Neyman, J.} (1971). Molecular studies of evolution: a source of novel statistical problems. 
In \textit{Statistical decision theory and related topics,} 
Gupta, S.S. and Yackel, J. (eds).
Academic Press,
New York, 1--27. 

\bibitem{Ostro37}
\textsc{Ostrowski, A.} (1937). \"Uber die Determinanten mit überwiegender Hauptdiagonale.
\textit{Commentarii Mathematici Helvetici} \textbf{10}, 69-96.

\bibitem{Ostro56}
\textsc{Ostrowski, A.} (1956).
Determinanten mit überwiegender Hauptdiagonale
und die absolute Konvergenz von linearen
Iterationsprozessen.
\textit{Commentarii Mathematici Helvetici} \textbf{29}
175 -- 210 


\bibitem{RobThay82}
\textsc{Rubin, D.B. and Thayer, D.T.} (1982). EM algorithms for ML factor analysis.
\textit{Psychometrika} \textbf{47}, 69--76.

\bibitem{SadLau13}
\textsc{Sadeghi, K. and Lauritzen, S.L. } (2014).
Markov properties for mixed  graphs.
\textit{Bernoulli} \textbf{20}, 395--1028.


\bibitem{Schur1917}
\textsc{Schur, I.} (1917). \"Uber  Potenzreihen, die im Innern des Einheitskreises beschr\"ankt sind. 
\textit{J. Reine Angew. Mathem.}
 \textbf{147}, 205--232.
 
 
\bibitem{Spearm1904} 
\textsc{Spearman C.} (1904).
General intelligence, objectively determined and
measured.
\textit{American Journal of Psychology} 
\textbf{15}, 201--293.

 
 
\bibitem{SpirGlySchein93}
\textsc{Spirtes, P., Glymour, C. and Scheines R.} (1993). 
\textit{Causation,
prediction and search}.  
Springer, New York. (2nd ed., 2001, MIT Press, Cambridge)


\bibitem{Stang97}
\textsc{Stanghellini, E.} (1997).  Identification of a single-factor model using graphical Gaussian rules.
\textit{Biometrika} \textbf{84}, 241--244.

 \bibitem{StangVan13}
\textsc{Stanghellini E. and Vantaggi, B.}  (2013)
Identification of discrete concentration graph models with one hidden binary variable.
\textit{Bernoulli}. \textbf{19}, 1920--1937. 

 \bibitem{Ark12}
 \textsc{Van der Ark, L.A.} (2012). New developments in Mokken scale analysis in R.
 \textit{J. Statist. Software} \textbf{48,5}

 

\bibitem{Wer12}
\textsc{Wermuth, N.}  (2012). Traceable regressions.
\textit{International Statistical Review}. \textbf{80}, 415--438.


\bibitem{WerCoxMar06}
\textsc{Wermuth, N., Cox, D.R. and Marchetti, G.M.} (2006).
Covariance chains. 
\textit{Bernoulli}, \textbf{12}, 841--862.

 
 \bibitem{WerMarCox09}
\textsc{Wermuth, N.,  Marchetti, G.M.  and Cox, D.R.} (2009).
Triangular systems for symmetric
binary variables.
\textit{Electr. J. Statist.}
\textbf{3}, 932--955.



\bibitem{WerSad12}
\textsc{Wermuth N. and Sadeghi, K.} (2012).
Sequences of regressions and their independences (with discussion).
\textit{TEST}
\textbf{21}, 215-279.


\bibitem{WerWieCox06}
\textsc{Wermuth, N.,  Wiedenbeck, M. and Cox, D.R.} (2006).
Partial inversion for linear systems and partial closure of independence
graphs.
\textit{BIT, Numerical Mathematics}
\textbf{46}, 883--901.


 \bibitem{Zentgraf75}
 \textsc{Zentgraf,  R.} (1975).
 A note on Lancaster's definition of higher-order interactions.
 \textit{Biometrika}
 \textbf{62}, 375--378.

 \bibitem{ZwierSmith11}
 \textsc{Zwiernik, P.  and Smith, J.Q.} (2011).
 Implicit inequality constraints in a
binary tree model.
 \textit{Electr. J. Statist.}
\textbf{5}, 1276--1312.
 
 

\end{thebibliography}
\end{document}